\documentclass{article}

\usepackage{amsfonts,amsmath}
\usepackage{amsmath}
\usepackage{amssymb}
\usepackage{amsthm}
\usepackage{algorithm}
\usepackage{enumitem}
\usepackage{algpseudocode} 
\usepackage{verbatim}
\usepackage{url}

\newtheorem{theorem}{Theorem}
\newtheorem{lemma}{Lemma}

\newtheorem{proposition}{Proposition}

\newcommand{\ket}[1]{\lvert #1 \rangle}
\newcommand{\bra}[1]{\langle #1 \rvert}

\newcommand{\size}[1]{\lvert #1 \rvert}

\newcommand{\C}{\mathbb{C}}
\DeclareMathOperator{\poly}{\mathrm{poly}}

\newcommand{\floor}[1]{\lfloor #1 \rfloor}

\newcommand{\suppress}[1]{}

\newcommand{\Z}{{\mathbb Z}}

\newcommand{\subgstate}[2]{\Xi_{#1,#2}}

\usepackage{xcolor}

\begin{document}

\title{Zero sum subsequences and hidden subgroups}

\author{
Muhammad Imran
\\
Department of Algebra,
\\
Budapest University of Technology and Economics,
\\
Egry J\'ozsef u. 1,
H-1111 Budapest, Hungary.
\\
E-mail: \texttt{muh.imran716@gmail.com}
\and
G\'abor Ivanyos
\\
Computer and Automation Research Institute,
\\
E\"otv\"os Lor\'and Research Network, 
\\
Kende u. 13-17, H-1111 Budapest, Hungary.
\\
E-mail: \texttt{Gabor.Ivanyos@sztaki.hu}
}

\maketitle

\begin{abstract}
We propose a method for solving the hidden subgroup problem in nilpotent groups. The main idea is iteratively transforming the hidden subgroup to its images in the quotient groups by the members of a central series, eventually to its image in the commutative quotient of the original group; and then using an abelian hidden subgroup algorithm to determine this image. Knowing this image allows one to descend to a proper subgroup unless the hidden subgroup is the full group. The transformation relies on finding zero sum subsequences of sufficiently large sequences of vectors over finite prime fields. We present a new deterministic polynomial time algorithm for the latter problem in the case when the size of the field is constant. The consequence is a polynomial time exact quantum algorithm for the hidden subgroup problem in nilpotent  groups having constant nilpotency class and whose order only have prime factors also bounded by a constant.
\end{abstract}


\paragraph*{Acknowledgments.} 
The research of the second author was supported by the Hungarian 
Ministry 
of Innovation and Technology 
NRDI Office within the framework of the Artificial
Intelligence National Laboratory Program.

\section{Introduction}\label{sec:intro}
The standard version of the hidden subgroup problem (HSP for short) 
is the following. Given a function $f$ on the group $G\rightarrow \{0,1\}^r$ with the property that there is a subgroup $H$ such that $f(x)=f(y)$ if and only if $x$ and $y$ are in the same left coset of $H$, find the subgroup $H$. Perhaps Kitaev was the first who observed that Shor's factoring and discrete logarithm algorithms can be generalized to solve the HSP in finite abelian groups (and also in certain infinite commutative groups) in polynomial time. Much less is known about the complexity of the problem in non-commutative groups. The most general result is due to Ettinger, Hoyer and Knill. They showed in~\cite{EHK04} that the {\em query complexity} of the problem in finite not necessarily abelian groups is polynomial. Regarding the time complexity, Kuperberg's subexponential time quantum algorithm~\cite{Kuperberg} for the HSP in dihedral and very similar groups is perhaps the best known result. It has a remarkable extension by Alagic, Moore and Russell~\cite{AlaMooRus}to a special HSP in a class including non-solvable groups. There are some classes of groups in which the HSP can be solved in polynomial time. See the survey papers by Lomont~\cite{LomontSurvey} and by Wang~\cite{WangSurvey} for early results of this kind. The paper~\cite{LomKaufTool} by Lomonaco and Kauffman proposes interesting derivatives and generalizations of the Shor-Kiteav algorithm. The paper~\cite{HorKah} by Horan and Kahrobaei discusses cryptographic aspects of the HSP and reports also on more recent results. The hidden shift problem in abelian groups (and hence the HSP in the related semidirect product groups) appears to be quite popular in post-quantum cryptography, see, e.g.,~\cite{CasVdm} by Castryck and Vander Meeren and~\cite{AlaRus} by Alagic and Russell. In~\cite{BaeLee}, Bae and Lee propose a polynomial time solution to a continuous version of the hidden shift problem.

A quantum procedure is exact if it returns a correct output (after a final measurement) with probability one. Besides that exact quantum algorithms can be considered as counterparts of deterministic classical methods, their measurement-free versions can serve as ingredients of larger unitary procedures. The method of~\cite{EHK04} has an exact version, so it is natural to ask that in which classes of groups can the HSP be solved by an exact quantum algorithm in polynomial time. Brassard and Hoyer~\cite{BraHoy} presented a polynomial time exact method that works in $\Z_2^n$. In~\cite{CaiQiu}, Cai and Qiu proposed a simpler efficient exact method for Simon's problem (a special, though arguably the hardest instance of the HSP in $\Z_2^n$). Efficient exact algorithms with optimal query complexity for the HSP in $\Z_2^n$ appeared independently in~\cite{BonnetainSimon} by Bonnetain and in~\cite{WuQiuTanLiCai} by Wu et al. Mosca and Zalka 
in~\cite{MosZal} proposed an efficient exact solution of the discrete logarithm problem in cyclic groups of known order. An exact quantum algorithm for the HSP in $\Z_{m^k}^n$ for general $m$ was presented recently in~\cite{II22}, settling the case of abelian groups under the assumption that 
a multiple of the prime factors of the order of the group is known. 

\smallskip

In this paper we present an approach to solving the hidden subgroup problem in nilpotent groups that have nilpotency class $O(1)$. Our main result is a polynomial time exact quantum algorithm for the HSP in such groups only having prime factors also of size $O(1)$ in their order. We assume that the group $G$ is given as a black-box group with unique encoding. The main strategy of our algorithm is essentially a reduction to instances of the hidden subgroup problem in quotient groups of subgroups of $G$. We choose an input model suitable for such a reduction. 

In the standard version, the input is given by an oracle which is a unitary map computing $\ket{x}\ket{f(x)}$ from $\ket{x}\ket{0}$. The usual hidden subgroup algorithms start with computing the superposition $\frac{1}{\sqrt{\size{G}}}\sum_{x\in G}\ket{x}\ket{f(x)}$ using the oracle and most of them ignore the second register that holds the value of $f$ and work with the coset superpositions $\ket{xH}=\frac{1}{\sqrt{\size{H}}}\sum_{y\in H}\ket{xy}$ in the sequel, see e.g.,~\cite{LomKaufTool}. These methods, as noted in~\cite{Kuperberg}, remain applicable in the context where the oracle is assumed to generate copies of a mixture of the coset superpositions. This holds in particular in the case of the exact abelian hidden subgroup algorithm of~\cite{II22}.

Specifically, we consider the state $\frac{1}{\sqrt{\size{G}}}\sum_{x\in G}\ket{x}\ket{f(x)}$ as a purification of the mixed state $\subgstate{G}{H}=\frac{1}{\size{G:H}}\sum_{x\in X}\ket{xH}\bra{xH}$, where $X$ is any left transversal of $H$ in $G$, in order to have a unitary oracle. The state $\subgstate{G}{H}$ is referred to as a (hidden) subgroup state. We assume that our hidden subgroup $H$ is given by a unitary map (referred as oracle) that, on zero input, returns a copy of an {\em arbitrary} (though fixed) purification of the subgroup state $\subgstate{G}{H}$. 

It will be convenient to introduce a subtask of the HSP, namely computing the hidden subgroup modulo the commutator subgroup of $G$, that is, the subgroup $HG'$ where $H$ is the hidden subgroup. We use the shorthand HSMC for this problem. To illustrate the power of HSMC in nilpotent groups note that it naturally includes the commutative case of the HSP and that having computed the subgroup $HG'$ and it is a proper subgroup of $G$ then we can descend to it to compute $H$, while if $HG'=G$ then $H=G$ because in a nilpotent group every maximal subgroup contains the commutator. 

We give a high-level description of a strategy for solving the problem HSMC in a class of nilpotent groups. We call a group $G$ {\em semi-elementary} if $G$ is a $p$-group for some prime $p$ such that $G/G'$ is elementary abelian. In a semi-elementary group $G$, our strategy for computing the hidden subgroup modulo the commutator is based on iterating the following procedure. Assume that $L$ is an elementary abelian subgroup contained in the center of $G$. Then we create a copy of the subgroup state corresponding to $HL/L$ in the quotient group $G/L$ from sufficiently many copies of the subgroup state for $H$ in $G$. We refer to this procedure (as well as some simpler ones) as {\em subgroup state conversion}. This conversion is based on finding zero sum subsequences of sufficiently long sequences of elements of $L$. Eventually, in $c-1$ rounds of iteration, where $c$ is the nilpotency class of $G$, we compute a copy of the subgroup state corresponding to $HG'/G'$ in $G$. (The semi-elementary property ensures the existence of a standard central series of length $c$ with elementary abelian factors.) Finally, from sufficiently many copies of such subgroup states we compute $HG'/G'$ using the exact abelian hidden subgroup algorithm of \cite{II22}. Fortunately, semi-elementary groups occur as factor groups of subgroups of nilpotent groups frequently enough to make a reduction from the HSP to the special case of HSMC possible, see Proposition~{\ref{prop:HSP-to-HSMC}} for details. The main result we obtain is the following.

\begin{theorem}\label{thm:HSP}
Suppose that $G$ is a nilpotent group of class bounded by a constant and that the prime factors of $\size{G}$ are also bounded by a constant. We assume that $G$ is a black-box group with unique encoding of elements by $\ell$-bit strings. Then there is an exact quantum algorithm that solves the hidden subgroup problem in $G$ using $\poly(\ell)$ operations and $\poly(\log\size{G})$ calls to the subgroup state creating oracle and its inverse.
\end{theorem}

\noindent
{\em Related results.}
In the exact setting,~\cite{II22} efficiently solves the abelian case without the restriction on the prime factors of $\size{G}$. There are quite a few related  non-exact polynomial-time algorithms. Among them, the result of \cite{FIMSS}, which solves the HSP in {\em solvable} groups that have derived series and exponent bounded by constants, is perhaps the closest to Theorem~\ref{thm:HSP}. This class of groups covers the groups for which our result is applicable, except those that have exponent divisible by large powers of small primes. Note however, that the case of these groups could be efficiently treated by a combination of the reduction of Proposition~\ref{prop:HSP-to-HSMC} with the algorithm of \cite{FIMSS}. We remark that the semidirect product group in which the HSP is equivalent to the hidden shift problem over $\Z_{2^k}^n$ is a nilpotent group of class $k$. Bonnetain and Naya-Plasencia~\cite{BonNaya} propose a non-exact method whose main ingredient can be considered as a combination of Kuperberg's sieve with 
finding zero sum subsequences in $\Z_2^n$ using linear algebra.

The case of nilpotency class at most two is efficiently treated by the non-exact method of~\cite{ISS12}, without any restriction on the size of the prime factors of $\size{G}$. It is worth mentioning that by technical content, \cite{ISS12} can be considered as the closest relative of the present paper. The idea of reducing the HSP to HSMC stems from there and many ingredients of the reduction appeared in that paper. Also, the key tool of~\cite{ISS12}, using several coset superpositions and the quantum Fourier transform of a central subgroup can be considered as some (though less transparent) form of subgroup state conversion. In the class two case, however, there is a more powerful tool to cancel out characters of the subgroup: one can also apply twists with certain nice automorphisms of the group that do not change the hidden subgroup too much. Unfortunately, such automorphisms do not exist in general nilpotent groups of class greater than two.

The methods of \cite{DHIS14,IS17} offer efficient solution to the HSP in certain nilpotent groups of higher class, again with potentially large prime factors in there orders. These groups have a normal subgroup with an abelian factor group of restricted kind (e.g., cyclic). These methods as well as that of \cite{FIMSS} are of highly non-exact nature. Probably, the technique of \cite{DHIS14} can be made exact with some efforts.

The {\em Davenport constant} $S(A)$ of a finite abelian group$A$ is the smallest number $s$ such that any sequence of $s$ elements of $A$ contains a nonempty subsequence adding up to the zero element of $A$. The name comes from that H.~Davenport proposed determining $S(A)$ in the case when $A$ is the ideal class group of a number field as a measure for non-uniqueness of factorization of the integers of the field. The general problem has become a famous question of additive combinatorics. Olson~\cite{Olson} determined the exact value of the Davenport constant of $p$-groups; in particular for $\Z_p^n$ it is $1+n(p-1)$. What we are looking for
is an "effective" Davenport constant: what is the smallest number $S'=S^{\cal B}(A)$ such that from any sequence of $S'$ elements of $A$, algorithm ${\cal B}$ finds a non-empty zero sum subsequence in time  polynomial in $S'\log\size{A}$ (roughly this is the bit size of the input sequence). In this paper we give a deterministic algorithm $\cal B$ running in time $\poly(n)$, that, for $p=O(1)$, given a sequence of $S^{\cal B}(\Z_p^n)=\poly(n)$ vectors from $\Z_p^n$, returns a zero sum subsequence. 

\smallskip

The structure of the rest of the paper is the following. In Section~\ref{sec:prelim}, we give some background
material on exact quantum procedures, on nilpotent black-box groups and on computations with  them, on (hidden) subgroups states and their purifications and present methods to convert subgroup states in the entire group to those in subgroups and - in certain very easy cases - in factor groups. Proposition~\ref{prop:HSP-to-HSMC}, the existence of an exact polynomial time reduction from the HSP in 
general nilpotent groups to the problem HSMC in semi-elementary groups is proved in Section~\ref{sec:HSP-to-HSMC}. Section~\ref{sec:convert} is devoted to converting several copies of a subgroup state in a semi-elementary group to a copy of a subgroup state in the abelian factor of the group. As an application of the technique, we prove Proposition~\ref{prop:mod_commutator} which tells us that we can solve by a polynomial time exact quantum algorithm the problem HSMC in a semi-elementary $p$-group of constant nilpotency class provided that we can find zero sum subsequences of sequences consisting of $\poly(n\log p)$ vectors from $\Z_p^n$ in time $\poly(n\log p)$. In Section~\ref{sec:zsum}, we prove Theorem~\ref{thm:zsum} on efficiently solvability the latter task in the case when $p$ is bounded by a constant. Propositions~\ref{prop:HSP-to-HSMC} and~\ref{prop:mod_commutator}, together with Theorem~\ref{thm:zsum}, immediately imply Theorem~\ref{thm:HSP}. Section~\ref{sec:conclusion} is devoted to concluding remarks.

\section{Preliminaries}\label{sec:prelim}

\subsection{On exact quantum computations}
To obtain sufficiently general intermediate results, we use the model of uniform circuit families described
by Nishimura and Ozawa~\cite{NishOzaUnif}. 
This is because some of the exact methods of \cite{II22}
as well as our main conversion technique
work under the assumption that the quantum Fourier transforms 
and their inverses modulo the prime factors
of $\size{G}$ can be exactly implemented. 
As it is pointed out in~\cite{NishOzaQTM}, this task cannot be 
accomplished using a fixed finite gate set. For the sake of transparency, 
we state our intermediate result using assumptions on availability of 
the quantum Fourier transforms rather than on gates 
required by the exact implementations of them. 
(See the implementation of the Fourier transform modulo general numbers
proposed by Mosca and Zalka~\cite{MosZal}.) 
Note however, that for the case of our Theorem~\ref{thm:HSP}, 
where these primes are assumed to be bounded by a constant, 
a constant number of gates are sufficient and hence, 
by~\cite{NishOzaPerfect}, the theorem remains valid 
in the quantum Turing machine model of Bernstein and Vazirani~\cite{BerVaz}.

\subsection{Groups}\label{subsec:groups}
For standard notations and concepts from group theory such as subgroups, normal subgroups, cosets, conjugates, commutators, commutator subgroup, center, etc., we refer the reader to the textbooks, e.g., to \cite{Robinson}. For subsets $U$ and $V$ of $G$ we denote by $UV$ the set $\{uv:u\in U,v\in V\}$. If both $U$ and $V$ are subgroups and either $U$ or $V$ is normal in $G$ then $UV$ is a subgroup. For subgroups $U,V$, by $[U,V]$ we denote the {\em subgroup generated by} the commutators $[u,v]$ ($u\in U,v\in V$). Recall that the lower central series of a finite group $G$ is the sequence $G=G_0>G_1>\ldots>G_c$ of normal subgroups $G_i\lhd G$ recursively defined as $G_i=[G,G_{i-1}]$. Here we assume that $c$ is the smallest index $i$ such that $G_i=[G,G_i]$. The group $G$ is nilpotent if $G_c=\{1\}$ and then $c$ is called the (nilpotency) class of $G$. A finite group is nilpotent if and only if it is the direct product of its Sylow subgroups.

To obtain sufficiently general results, we work over {\em black-box} groups with unique encoding of elements.
The concept captures various "real" groups such as permutation groups and matrix groups over finite fields. 
Elements of a black-box group are represented by binary strings of a certain length $\ell$ and the group operations  are given by oracles and as input, a generating set for the group is given. Subgroups will also be given by sets of generators. One can use the exact polynomial time quantum membership test of~\cite{II22}
to reduce the size of generating sets to at most $\log\size{G}$.

During the rest of this part, we assume that $G$ is a nilpotent black-box group of class $c$ and the prime factors of $\size{G}$ are known. 

For a normal subgroup $N$ of $G$, the subgroup $[G,N]$ is a normal subgroup of $G$ contained in $N$. 
If $\Gamma$ and $\Delta$ are sets of generators for $G$ and $N$, respectively, a generating set for $[G,N]$  can be obtained by taking the commutators $[x,y]$ for $x\in \Gamma,y\in \Delta$ and then adding iterated commutators with elements of $\Gamma$ until the subgroup generated by the elements stabilizes. For testing stabilization, one can use the exact quantum subgroup membership algorithm of~\cite{II22}. This gives a polynomial time exact method in particular to compute the lower central series. 

Below we describe efficient solutions to some further group theoretic tasks that we use in our hidden subgroup algorithm. The $p$-Sylow subgroup of $G$ can be computed as follows. Let $\Gamma$ be a generating set for $G$. Then for each $g\in\Gamma$ we compute the order $o_g$ of $g$ and decompose $o_g$ as the product
$p^\alpha {o_g}'$ where ${o_g}'$ is coprime with $p$. Then the $g^{{o_g}'}$ ($g\in \Gamma$) generate the (unique) $p$-Sylow subgroup of $G$. We shall compute hidden subgroups in $G$ by computing the intersections with the Sylow subgroups.

The normalizer of a subgroup of $G$ can be computed using the deterministic polynomial method of Kantor and Luks~\cite{KantorLux}. It was originally described for nilpotent permutation groups but it also finds normalizers in any nilpotent black box group of order having small prime factors only.

Assume that $L$ is a subgroup of $G$. It will be useful to decompose elements $x$ of $G$ as products of the form $\alpha_L(x)\beta_L(x)$ where $\beta_L(x)\in L$ and $\alpha_L(x)$ depend only on the coset $xL$. (Thus the range of $\alpha_L$ is a transversal of $L$ in $G$.) To this end, compute a chief series (a series of normal subgroups with cyclic factors of prime order) $G=K_0>K_1>\ldots>K_r=1$. Perhaps the easiest way to obtain such series is taking  a refinement of the lower central series. By taking the subgroups $K_iL$,
and removing repeated elements, we obtain a subnormal series $G=M_0>M_1>\ldots>M_s=L$ with cyclic factors of prime order. Also take elements $a_i\in M_{i-1}\setminus M_i$ and denote by $p_i$ the order of ${M_{i-1}/M_i}$
($i=1,\ldots,s$). Then the elements $a_1^{\gamma_1}a_2^{\gamma_2}\ldots a_s^{\gamma_s}$ ($(\gamma_1,\ldots\gamma_s)\in \prod_{i=1}^s\Z_{p_i}$) are a left transversal of $L$ in $G$. For an element $x\in G$, the representative of the coset in this transversal can be computed as follows. First we find the smallest non-negative integer $\gamma_1$ such that $xa_1^{-\gamma_1}\in M_1$ by computing the base $a_1$ discrete logarithm of $x$ modulo $M_1$. This can be done by solving an instance of the hidden subgroup problem in $\Z_{p_1}^2$. Specifically, we define the function $(\beta,\gamma)\mapsto\ket{x^\gamma}\ket{a_1^{-\beta}M_1}$. The function can be evaluated with the aid of computing the uniform superposition $\ket{M_1}$ using the exact version \cite{II22} of Watrous's method~\cite{Watrous}. The values are $p$ pairwise orthogonal states and the hidden subgroup is $\{(\delta,\gamma):x^\delta a_1^{-\gamma}\in M_1\}$. We use the exact hidden subgroup algorithm of~\cite{II22} to find a generator of this group. From this, $\gamma_1$
can be obtained in an obvious way. Now we proceed with $xa_1^{-\gamma_1}$ to compute $\gamma_2$, and so on.
We set $\alpha_L(x)=a_1^{\gamma_1}\ldots a_r^{\gamma_r}$ and $\beta_L(x)=\alpha_L(x)^{-1}x$.

If $L$ is a normal subgroup of $G$, we can encode the coset $xL$ by $\alpha_L(x)$. This makes the factor group $G/L$ a black-box group: the elements are encoded by the elements of the transversal $\{\alpha(x);x\in G\}$ and the multiplication oracle is obtained as a composition of the multiplication oracle for $G$ with the computation of the function $\alpha_L$.

\subsection{Subgroup states and purifications} \label{subsec:subgstates}
Let $G$ be a finite group and let $H$ be a subgroup of $G$. We consider elements of the group algebra $\C G$ as pure quantum states. (The "natural" scalar product $(\sum_x\alpha_x \ket{x}, \sum_y\alpha_y \ket{y})=\sum\alpha{\overline \beta}xy^{-1}$ makes $\C G$ a Hilbert space where the group elements form an orthonormal basis.) 

A (left) coset superposition of $H$ in $G$ is the uniform superposition $\ket{aH}=\frac{1}{\sqrt{\size{H}}}\sum_{h\in H}\ket{ah}$ where $a\in G$. The (left) subgroup state of $H$ in $G$ is the mixed state with the density matrix $$\subgstate{G}{H}=\frac{1}{\size{G}}\sum_{a\in G}\ket{aH}\bra{aH} =
\frac{1}{\size{G:H}}\sum_{a\in X}\ket{aH}\bra{aH},$$ where $X$ is any left transversal (a set of representatives of the left cosets) of $H$ in $G$. 

A {\em purification} of $\subgstate{G}{H}$ is any pure state $\ket{\psi}\in \C G\otimes V$ for some Hilbert space $V$ such that $\subgstate{G}{H}$ is the relative trace of $\ket{\psi}\bra{\psi}$ with respect to the second subsystem. For general facts about purification of mixed states, in particular for the connection with Schmidt decompositions, we refer the reader to Section~{2.5} of~\cite{NiChu}. The following lemma gives a characterization of purifications of subgroup states.

\begin{lemma}\label{lem:purification}
The pure state $\ket{\psi}\in \C G\otimes V$ is a purification of the subgroup state $\subgstate{G}{H}$ if and only if it can be written as $$\ket{\psi}=\frac{1}{\sqrt{\size{G}}}\sum_{x\in G}\ket{x}\ket{v(x)},$$
where the states $\ket{v(x)}$ and $\ket{v(y)}$ are equal if $x$ and $y$ are in the same left coset of $H$ and orthogonal otherwise.
\end{lemma}

\begin{proof}
The "if" part follows easily from that the conditions on $\ket{v()}$ imply $\ket{\psi}=\frac{1}{\sqrt{k}}\sum_{a\in X}\ket{aH}\ket{v(a)}$.

To see the the "only if" part, recall that a Schmidt decomposition of a state $\ket{\psi}\in \C G\otimes V$ is of the form $\ket{\psi}=\sum_{i=1}^m\lambda_i\ket{u_i}\ket{v_i}$ where $m=\size{G}$, $\ket{u_1},\ldots,\ket{u_m}$ is an {\em arbitrary} orthonormal basis of $\C G$ in which the relative trace of $\ket{\psi}\bra{\psi}$ w.r.t.~to the second subsystem is diagonal (with entries $\lambda_1,\ldots,\lambda_m$) 
and the system of the vectors $v_i$ corresponding to nonzero eigenvalues $\lambda_i$ is an orthonormal system of vectors in $V$. The vectors $v_i$ depend on the choice of the basis $\ket{u_i}$ ($i=1,\ldots,m$). Notice 
that the only nonzero eigenvalue of $\subgstate{G}{H}$ is $\frac{1}{k}$ with multiplicity $k$, where $k=\size{G:H}$. The coset superpositions give an orthonormal basis of the corresponding eigenspace. Thus, if $\ket{\psi}$ is a purification of $\subgstate{G}{H}$ then a Schmidt decomposition of $\ket{\psi}$ which is a purification of $\subgstate{G}{H}$ is of the form $\ket{\psi}=\frac{1}{\sqrt{k}}\sum_{i=1}^k\ket{u_i}\ket{v_i}$ where $\ket{u_1},\ldots,\ket{u_k}$ is an {arbitrary} orthonormal basis of the $\frac{1}{k}$-eigenspace of $\subgstate{G}{H}$ and $\ket{v_1},\ldots,\ket{v_k}$
is an orthonormal system of $V$.  In particular, if $X=\{a_1,\ldots,a_k\}$ then by taking $\ket{u_i}=\ket{a_iH}$ and by defining $\ket{v(x)}=\ket{v_i}$ for $x\in a_iH$, we obtain
$\ket{\psi}=\frac{1}{\sqrt{k}}\sum_{a\in X}\ket{aH}\ket{v(a)}=\frac{1}{\sqrt{\size{G}}}\sum_{x\in G}\ket{x}\ket{v(x)}$.
\end{proof}

\subsection{Basic subgroup state conversions}\label{subsec:basic_conv}
Given a subgroup $L$ of $G$, a copy of (a purification of) the subgroup state $\subgstate{G}{H}$ can be "converted" to a copy of (a purification of) $\subgstate{L}{H\cap L}$ by replacing $\ket{x}$ with the decomposition $\ket{\beta_L(x)}\ket{\alpha_L(x)}$ obtained by the method outlined in Subsection~\ref{subsec:groups} for $x\in G$, and "ignoring" $\ket{\alpha_L(x)}$ (passing this part to the purifying subsystem). To see this, let $\frac{1}{\sqrt{\size{G}}}\sum_{x\in G}\ket{x}\ket{\psi(x)}$ be a purification of $\subgstate{G}{H}$, with $\ket{\psi(x)}$ and $\ket{\psi(y)}$ are equal if and only if  $y^{-1}x\in H$, and orthogonal otherwise.  Then, the substitution gives the state $\frac{1}{\sqrt{\size{L}}}\sum_{x\in L}\ket{x}\frac{1}{\sqrt{\size{G:L}}}\sum_{y\in Y}\ket{y}\ket{\psi(yx)}$ where $Y=\{\alpha_L(z):z\in G\}$. Now if $x_1,x_2\in L$ are from the same left coset of $H\cap L$ then $\ket{\psi(yx_1)}=\ket{\psi(yx_2)}$ for every $y\in Y$ and hence the states $\frac{1}{\sqrt{\size{G:L}}}\sum_{y\in Y}\ket{y}\ket{\psi(yx_i)}$ are equal ($i=1,2$), while otherwise they do not overlap as for $y_1,y_2\in Y$ either  $\ket{y_1}$ and $\ket{y_2}$ are orthogonal or (for $y_1=y_2$) $\ket{\psi(y_1x_1)}$ and $\ket{\psi(y_1x_2)}$ are orthogonal. We shall refer to this procedure as {\em restriction}. The term is justified by that in the standard version of the HSP, one could obtain an instance of the HSP in the subgroup $L$ by restricting the ``hiding function'' to $L$.

Similarly, assume that $L$ is a normal subgroup of $G$ contained in $H$. Then a copy of (a purification of) 
the subgroup sate $\subgstate{G}{H}$ can be converted to a copy of  (a purification of) $\subgstate{G/L}{H/L}$ by replacing $x$ with $\ket{\alpha_L(x)}\ket{\beta_L(x)}$ and passing $\ket{\beta_L(x)}$ to the purifying subsystem. This corresponds to the technique called "pushing" in \cite{LomKaufGrover,LomKaufTool}.

\section{A group-theoretic reduction}\label{sec:HSP-to-HSMC}
In this section we prove the following.

\begin{proposition}\label{prop:HSP-to-HSMC}
Let $G$ be a nilpotent black-box group of class at most $c$ and assume that the prime factors of $\size{G}$ are given as part of the input and that for each such prime $p$ the quantum Fourier transform modulo a multiple of $p$ and its inverse can be implemented by an efficient exact quantum procedure. Then, the HSP in $G$ can be reduced by an exact procedure in time $\poly(\log \ell)$ to $\poly(\log \size{G})$ instances of the problem HSMC in semi-elementary quotient groups of subgroups of $G$. (The elements of $G$ are assumed to be uniquely encoded by strings of length $\ell$.)
\end{proposition}

\begin{proof}
A finite nilpotent group $G$ is the direct product of its Sylow subgroups. Therefore any subgroup $H$ is the product of its Sylow subgroups. The $p$-Sylow subgroup of $H$ is $P\cap H$ where $P$ is the $p$-Sylow subgroup of $G$. The Sylow subgroups of $G$ can be computed using the method outlined in Subsection~\ref{subsec:groups}. One can convert subgroup states in $G$ to subgroup states in $P$ using restriction, see Subsection~\ref{subsec:basic_conv}.

In the rest of the description of the reduction we assume that $G$ is a $p$-group. We maintain a subgroup $H_0$ of $H$. Initially $H_0=\{1_G\}$. In each round of an outer loop of the algorithm $H_0$ will be increased if $H_0<H$. If $H_0$ is already $G$ then we can obviously stop. We will also maintain a subgroup $K$ 
of $G$ such that if $H_0<H$ then even $H_0<K\cap H$. Initially $K=N_G(H_0)$. This is a good choice because in a nilpotent group every proper subgroup has a strictly larger normalizer, therefore if $H_0<H$ then $H_0<N_H(H_0)=H\cap N_G(H_0)$. In an inner loop $K$ will be decreased until either $H_0$ is increased or $K$ becomes identical with $H_0$. In the latter case we can conclude that $H=H_0$ and stop the whole procedure. If the abelian factor $K/(K'H_0)$ is not elementary then we can replace $K$ with a proper subgroup as follows. Let $L>K'H_0$ be the subgroup of $K$ such that $L/(K'H_0)$ contains all the elements of order $p$ of $K/(K'H_0)$.  To compute $L$, first compute $K'$ and $K'H_0$. Then take the set of generators $\Gamma$ for $K$ and for each element $g\in \Gamma$, compute the smallest {\em positive} integer $\alpha_g$ such that $g^{p^{\alpha_g}}\in K'H_0$. The elements $g^{p^{\alpha_g-1}}$ ($g\in \Gamma$) generate $L$. If $L$ is a proper subgroup of $K$ then we replace $K$ with $L$ and repeat the step above. (Correctness of this is justified by observing that $L/H_0$ contains all the elements of order $p$ of $K/H_0$, whence if $H\cap K>H_0$ then also $H\cap L>H_0$.) Otherwise we have achieved that $K/(K'H_0)$ is elementary abelian. Then we compute $(H\cap K)K'/H_0$ using HSMC. If $(H\cap K)K'=K$ then $H\cap K=K$ because $K'$ is contained in every maximal subgroup of $K$. Then we can increase $H_0$ by replacing $H_0$ with $K$ and continue the outer loop. If $(H\cap K)K'<K$ we can replace $K$ with $(H\cap K)K'$ and continue the inner loop. 

Based on the descriptions above, we summarize the exact algorithm in the pseudocode below.

\begin{algorithm}[H]
\label{alg1}
\caption{Reduction to HSMC}
\scriptsize
\begin{algorithmic}[1]
\State \textbf{Initialize:} $H_0 \gets 1_G$;
\While{$H_0 < G$}
    \State{$K \gets N_G(H_0)$;}
    \State{${Found}\gets \mbox{False}$;}
    \While{${Found}=\mbox{False}$}
        \If{$K/(K'H_0)$ is elementery}
            \State{Use HSMC to compute $(H\cap K)K'/H_0$;}
            \If{$(H\cap K)K'=K$}
                \State{$H_0\gets K$;}
                \State{${Found}\gets \mbox{True}$;}
            \Else
                \State{$K\gets (H\cap K)K'$;}
                \If{$K=H_0$}
                    \State{\textbf{return} $H=K$.}
                \EndIf
            \EndIf
        \Else
            \State{For each $g\in \Gamma_K$ compute the smallest positive integer $\alpha_g$ with $g^{p^{\alpha_g}}\in K'H_0$;}
            \State{Compute $L=\langle g^{p^{\alpha_g-1}} ~|~ g\in \Gamma_K \rangle$;}
            \State{$K\gets L$;}
        \EndIf
    \EndWhile
\EndWhile
\end{algorithmic}
\end{algorithm}

If $\size{G}=p^n$ then the outer loop is executed at most $n$ times while within each round of the outer loop the inner loop has at most $n$ rounds. Thus we need at most $n^2$ calls to the HSMC procedure for factors
of subgroups of $G$ and further $n^2\poly(\ell)$ group and other operations. Note that all the groups we need to apply the HSMC procedure are of class at most $c$ because the family of nilpotent groups of class at most $c$ is closed under taking subgroups and factor groups.
\end{proof}

\section{The main conversion}\label{sec:convert}
Let $L$ be a subgroup of the center of $G$
isomorphic to $\Z_p^n$ where $p$ is a prime.
Then $L$ is a normal subgroup of $G$. Our aim
is to convert a copy of the subgroup state
$\subgstate{G}{H}$ to a copy of $\subgstate{G/L}{HL/L}$.
In the light of the second conversion ("pushing")
described in 
Subsection~\ref{subsec:basic_conv}, one could do it
by converting first to a copy of $\subgstate{G}{HL}$.

To this end, it would be desirable to have a procedure
that converts the coset superposition $\ket{aH}$ to
$\ket{aHL}$. A possible approach would be
computing $\ket{L}
=\frac{1}{\sqrt{\size{L}}}\sum_{z\in L}\ket{z}$
in a new register, multiplying $\ket{aH}$ with it
to obtain 
$\frac{1}{\sqrt{\size{HL}}}\sum_{z\in L}\sum_{x\in H}\ket{azx}\ket{z}$
and then trying to "disentangle" $\ket{z}$ from $\ket{azx}$. The
quantum Fourier transform of $L$ almost does this job: if we apply
it to the second register, we obtain the state
$$\frac{1}{\sqrt{\size{L}}}\sum_{y\in L}
\frac{\omega^{(y,z)}}{\sqrt{\size{HL}}}\sum_{z\in L}\sum_{x\in
H}\ket{azx}\ket{y},$$
where $\omega=e^{\frac{2\pi i}{p}}$ and
by $(,)$ we denote the standard scalar product of $L$
modulo $p$. For $y\in L$, let us denote by $P_y$
the linear transformation of $\C G$ mapping 
$\ket{x}$ to $\frac{1}{\sqrt{\size{L}}}\sum_{z\in L}\omega^{(y,z)}\ket{xz}$.
With this notation, the state we have can be rewritten as
$$\frac{1}{\sqrt{\size{L}}}
\sum_{y\in L}\ket{P_y(aH)}\ket{y}.$$
Using the assumption that $L$ is in the center of $G$,
a direct calculation shows that 
for every $x_1,x_2\in G$, we have
\begin{equation}
\label{eq:central}
\ket{P_y(x_1x_2)}=\ket{x_1P_y(x_2)}=\ket{(P_y(x_1)x_2)}.
\end{equation}

It is also straightforward to see that for every $x\in G$
and for every
$w\in L$, we have 
\begin{equation}
\label{eq:eigen}
\ket{w P_y(x)}=\ket{P_y(x)w}=\omega^{-(y,w)}\ket{P_y(x)}.
\end{equation}

We define the support of an element $\ket{u}$ of $\C G$ as the set
of elements appearing with nonzero coefficient in the
decomposition of $\ket{u}$ as a linear combination of group elements.
Using equality~(\ref{eq:central}), one can show that
if $x_1$ and $x_2$ are not in the same left coset of $LH$
then the states $\ket{P_y(aH)x_1^{-1}}$ and $\ket{P_y(aH)x_2^{-1}}$
are orthogonal. This is because the 
support of $\ket{P_y(aH)x_i^{-1}}$ is contained in
$LHx_i^{-1}=(x_iLH)^{-1}$ ($i=1,2$).
On the other hand, if $x_1H=x_2H$, then $Hx_1^{-1}=Hx_2^{-1}$ and
the two states are equal. 
By the characterization given in Lemma~\ref{lem:purification},
it follows that for any left coset $aH$, the
state $\frac{1}{\sqrt{\size{G}}}\sum_{x\in G}\ket{x}\ket{P_0(aHx^{-1})}$
is a purification of $\subgstate{G}{LH}$.

Of course it is hopeless to enforce $y=0$ in $\ket{P_y(aH)x_1^{-1}}$.
However, we can compute a state with essentially the same effect using
several copies of the subgroup state and by applying
an algorithm that finds
zero sum subsequences of sufficiently long sequences of elements of $L$.
Assume that we have a procedure that, for some $S=S(p,n)$, 
given an element ${\underline y}=(y_1,\ldots,y_S)\in L^S$ 
computes a non-empty subset $J({\underline y})$ of $\{1,\ldots,S\}$
such that
$\sum_{j\in J(\underline y)}y_j=0$.

Then, for a sequence $\ket{a_1 H}\ldots\ket{a_S H}$ we 
compute first 
$${\size L}^{-S/2}\sum_{{\underline y}\in L^S}
\ket{{\underline y}}\ket{P_{y_1}(a_1 H)}\ldots \ket{P_{y_S}(a_S H)}$$
by applying the Fourier method outlined above component-wise.
We next compute $\frac{1}{\sqrt{\size{G}}}\sum_{x\in G}\ket{x}$
in a fresh register and multiply by $x^{-1}$ the $j$th component
of $\ket{P_{y_1}a_1 H}\ldots \ket{P_{y_S}a_S H}$ if $j\in 
J({\underline y})$. Let $\chi_{\underline y}:\{1,\ldots,S\}\rightarrow
\{0,1\}$ denote the characteristic function of $J({\underline y})$.
Then the state we obtained
is $\frac{1}{\sqrt{\size{G}}}\sum_{x\in G}\ket{x}\ket{\psi(x)}$,
where $$\ket{\psi(x)}=
{\size L}^{-S/2}\sum_{{\underline y}\in L^S}
\ket{{\underline y}}\ket{P_{y_1}a_1 Hx^{-\chi_{\underline y}(1)}}
\ldots \ket{P_{y_S}a_S Hx^{-\chi_{\underline y}(S)}}.$$

Consider the term of $\ket{\psi(x)}$ corresponding to
any ${\underline y}$. As $J({\underline y})$ is non-empty,
we have that for $x_1,x_2$ not in the same left coset
of $LH$, the appropriate terms of $\ket{\psi(x_i)}$ are
orthogonal. As $\ket{\underline y}$ also appears in
the corresponding term, we have that $\ket{\psi(x_i)}$
are also orthogonal. On the other hand,
if $x_1,x_2$ are in the same left coset of $H$
then these states are equal term by term. Finally,
for $x\in L$, 
by (\ref{eq:eigen},
the term for ${\underline y}$ 
only gets a phase change by $\prod_{j\in J}({\underline y})\omega^{-(y_j,x)}=
\omega^{\sum_{j\in J({\underline y})}(y_j,x)}=\omega^0=1$
by the choice of $J({\underline y})$. It follows that
if $x_1$ and $x_2$ are in the same left coset of $LH$,
$\ket{\phi(x_1)}=\ket{\phi(x_2}$. Thus our 
state
is a purification of $\subgstate{G}{LH}$. As this
holds for any fixed $S$-tuple of left cosets of $H$, 
by linearity we also obtain 
a purification of $\subgstate{G}{LH}$ if we apply
the procedure to copies of a purification of $\subgstate{G}{H}$.
We obtained the following.

\begin{lemma}
\label{lem:main_convert}
Assume that we have an exact quantum procedure
(e.g., a deterministic polynomial time algorithm)
 that,
given any sequence $y_1,\ldots,y_S$ of $S=S(p,n)$ elements
of $\Z_p^n$, in time $T(p,n)\geq S(p,n)$
finds a non-empty subset of $\{1,\ldots,S\}$
such that $\sum_{j\in J}y_j=0$. Then we have an exact quantum procedure
using $n$ quantum Fourier transforms modulo $p$
that converts $S(p,n)$ copies of (a purification) of
$\subgstate{G}{H}$ to a copy of (a purification of)
$\subgstate{G/L}{HL/L}$ where $L$ is subgroup of the
center of $G$ isomorphic to $\Z_p^n$ in time 
$T(p,n)\poly(\log {\size{G}})$.
\end{lemma}

Using the lemma in iteration and applying
the exact abelian hidden subgroup algorithm
of~\cite{II22}, we can derive the following.

\begin{proposition}
\label{prop:mod_commutator}
Let $G$ be a semi-elementary 
black-box group with unique encoding of order $p^n$.
Assume that 
the quantum Fourier transform modulo $p$ and its inverse
can be implemented by an efficient exact algorithm and that,
like in Lemma~\ref{lem:main_convert}, we have
an exact method to find zero sum subsequences of sequences of
$S(p,n)$ elements of $\Z_p^n$ in time $T(p,n)\geq S(p,n)$. 
Then the problem HSMC can be solved by an exact quantum
algorithm that uses 
$\poly(T(p,n)^{O(c)}\ell)$
elementary operations, 
$\poly(T(p,n)^{O(c)}\log{\size{G}})$
applications of the group oracle, 
calls to the oracle computing the purification of the subgroup state;
and the inverses of these.
(The elements of $G$ are assumed to be uniquely encoded
by strings of length $\ell$.)
\end{proposition}

\begin{proof}
We compute the lower central series $G=G_0>G_1>\ldots>G_c=\{1\}$
using the method presented in Subsection~\ref{subsec:groups}. 
As $G/G'$ is elementary abelian, so are the factors
 $G_{i-1}/G_i$ ($i=1,\ldots,c$). This is because the factor groups
$G_{i-1}/G_i$ are homomorphic images of tensor powers (as $\Z$-modules)
of the $G/G'$, see Theorem~5.2.5 of~\cite{Robinson}. Also, isomorphisms
of $G_{i-1}/G_i$ with $\Z_p^{n_i}$ can be efficiently computed using the 
method of \cite{II22}.
Iteration of Lemma~\ref{lem:main_convert} gives a procedure to
convert $\prod_{i=2}^{c}S(n_i)$ copies of a purification
of $\subgstate{G}{H}$ to a copy of a purification of
$\subgstate{G/G'}{HG'/G'}$. The composition of instances of
the original subgroup state creating procedure (the calls to the
oracle) with the conversion gives
a procedure for
creating a purification of $\subgstate{G/G'}{HG'/G'}$. We can use
this as the oracle input for the exact 
hidden subgroup algorithm of~\cite{II22} in $\Z_p^{n_1}$.
For $i=1,\ldots,c$, we have $S(p,n_i)\leq S(p,n)$ and $T(p,n_i)\leq T(p,n)$
because  $\Z_p^{n_i}$ can be embedded in $\Z_p^n$ as a subgroup.
\end{proof}

In the non-exact setting, essentially 
the same proof gives the following.

\begin{proposition}
\label{prop:mod_comm_nonexact}
Let $G$ be a semi-elementary 
black-box group with unique encoding of order $p^n$.
Assume that there exists
a quantum (or a randomized) 
algorithm that finds zero sum subsequences of sequences of
$S(p,n)$ elements of $\Z_{p^n}$ in time $T(p,n)\geq S(p,n)$
with high probability. 
Then the problem HSMC can be solved by a quantum
algorithm that uses 
$\poly(T(p,n)^{O(c)}\ell)$
elementary operations, 
$\poly(T(p,n)^{O(c)}\log{\size{G}})$
applications of the group oracle and  
calls to the oracle computing the purification of the subgroup state.
\end{proposition}

\section{Zero sum subsequences in $\Z_p^n$}
\label{sec:zsum}

In this section, we assume that our input
is a sequence
of vectors from $\Z_p^n$. We also assume that
$p$ is an odd prime as for $p=2$ a zero sum
subsequence can be obtained from $n+1$ vectors
in the form of a zero linear combination.
As subsequences 
can be represented as subsets of the index set,
it will not be too
misleading to use the term
(sub)set  for a (sub)sequence.
Our strategy will be finding $p$ pairwise
disjoint subsets of input vectors having
equal sums. We will achieve this goal
by designing a method for finding a nontrivial
pair of subsets having equal sum
and then, like in~\cite{IS17},
applying the algorithm recursively to
obtain $4$, $8$, $16$, etc.~disjoint subsets with equal
sum. 

Note that a pair of disjoint subsets with equal sum
can be interpreted as a 
representation of the zero vector by a linear combination of
the input vectors 
with nonzero coefficients $1$ or $-1$ only. 
Based on this, it will be convenient to use the term
{\em signed subsets} and signed subset sums.
A signed subset of a set $S$ of vectors
is formally a function from $S$ to the set
$\{0,1,-1\}$. The support of such a signed subset
is the set of elements on which the function
takes nonzero values. With some sloppiness, we use the 
term {\em signed subset sum} to refer both to the signed subset
and to the value of the signed sum. (Technically,
a signed subset sum could be a data structure consisting of
the description of the signed subset and the value.) 
We
call two or more subset sums disjoint if their 
supports are pairwise disjoint. Based on the observation
that a signed subset sum of vectors that are
results of pairwise disjoint subset sums is again a signed subset sum
of the original vectors, 
one can build signed subset
sums hierarchically from smaller disjoint signed subset sums.
The trivial subset sum corresponds to the empty set 
with the zero vector as value.

A
 {\em linear relation} (or just relation for short)
among a collection of
vectors is an array of coefficients such that
the corresponding linear combination
is the zero vector.
It is often useful to omit
the vectors to which coefficient zero are assigned.
By taking the signed subsets of the vectors
having the same or the opposite coefficient
in a linear relation, we obtain
a linear relation among pairwise disjoint
signed subset
sums in which the coefficients are form
$\{1,\ldots,\frac{p-1}{2}\}$ and each coefficient
appears at most once. We call such a relation
of signed subset sums {\em standard}. 

We shall build standard linear relations
among signed subset sums with smaller and smaller
coefficients (among increasingly larger subset sums). 
The key idea is constructing first $\frac{(p-1)^2}{4}$
pairwise signed subset sums arranged in a square matrix 
having a relation in each row as well as in each column and subtracting
the sum of higher half of "horizontal" relations from the
sum of
the higher half of the "vertical" relations to obtain
a relation with coefficients between $1$ and $\frac{p-1}{4}$,
and iterating the construction. 
We give the details
in the following lemma and its proof.
(We present a version that even saves up
maintaining the first half of vertical relations.)

\begin{lemma}
\label{lem:halving}
Let $d$ be a positive integer. Assume that there is a deterministic
procedure $\cal A$ that, given $h(d,n)$ vectors from $\Z_p^n$, in 
time $\poly(h(d,n)\log p)$ finds $d$ pairwise disjoint signed subset 
sums $v_1,\ldots,v_{d}$ of the input vectors, not all empty, 
such that $\sum_{i=1}^{d}iv_i=0$. Then there also exists a deterministic 
procedure that, given $h(d,n)h(d,\lceil d/2 \rceil n)$ vectors,
in time $\poly(h(d,n)h(d,\lceil d/2 \rceil n)\log p)$
finds pairwise disjoint signed subset sums $w'_1,\ldots,w'_{\lfloor d/2\rfloor}$,
not all empty, such that $\sum_{i=1}^{\lfloor d/2\rfloor}iw'_i=0$. 
\end{lemma}

\begin{proof}
We divide the input set into $h(d,\lceil d/2\rceil n)$ pairwise disjoint parts
of size $h(d,n)$. We apply procedure $\cal A$ within each part.
This way for each 
$k=1,\ldots,h(d,\lceil d/2\rceil n)$, 
we get 
$d$ pairwise disjoint subset sums $u_{k1},\ldots,u_{kd}$,
not all empty, such that $\sum_{j=1}^{d} ju_{kj}=0$. 
For each $k$ we consider the concatenation $u_k$ of the vectors
$u_{kj}$ ($j=\lfloor d/2 \rfloor +1,\ldots,d$). These are vectors of dimension $
\lceil d/2 \rceil n$.
We apply procedure $\cal A$ to find pairwise disjoint signed subsets 
$M_{1},\ldots,M_{d}$ such that 
$\sum_{i=1}^{d} iu_i'=0$, where $u_i'$ is the signed
sum of the $u_k$s corresponding to the signed subset $M_i$.
Now for each $1\leq i\leq d$, $u_i'$ is the concatenation
of vectors $w_{ij}$ ($j=\lfloor d/2\rfloor+1,\ldots,d$). Here,
for $1\leq i,j\leq d$, $w_{ij}$ stands for the
signed subset sum obtained by joining the signed subset sums $u_{kj}$ according
to the signed subset $M_i$. The signed subset sums $w_{ij}$ are
pairwise disjoint, not all of them are empty and they satisfy the relations
$$\sum_{i=1}^{d} iw_{ij}=0\;\; (j=\lfloor d/2\rfloor+1,\ldots,d)$$
and
$$\sum_{j=1}^{d} jw_{ij}=0\;\; (i=1,\ldots,d).$$
We subtract the sum of the last $\lceil d/2\rceil$ ("horizontal")
relations
of the second kind from the sum the 
$\lceil d/2\rceil$ ("vertical") relations of the first kind and obtain
the relation
$$\sum_{i=1}^{\lfloor d/2 \rfloor}
\sum_{j=\lfloor d/2 \rfloor+1}^{d} iw_{ij}
-
\sum_{i=\lfloor d/2 \rfloor+1}^{d}
\sum_{j=1}^{\lfloor d/2 \rfloor}jw_{ij}
+\sum_{i,j=\lfloor d/2 \rfloor+1}^{d}(i-j)w_{ij}=0. 
$$

Notice that for $\lfloor d/2 \rfloor+1\leq i,j\leq d$,
we have $|i-j|\leq \lfloor d/2 \rfloor$. Therefore,
by flipping signs where appropriate and then joining
the signed subsets with equal coefficients,
we obtain pairwise disjoint subset sums $w'_1,\ldots,w'_{\lfloor d/2 \rfloor}$
with $\sum_{i=1}^{\lfloor d/2 \rfloor}iw'_i=0$.

The subset sums $w_i'$ can all
be empty only if each $w_{ij}$
is empty
when $i\neq j$ and at least one of $i$ and $j$
is greater than $\lfloor d/2 \rfloor$. Assume that this is
the case. Then, if there is an index $i>\lfloor d/2 \rfloor$
such that $w_{ii}$
is non-empty then $w_{ii}$ 
must be
itself a nontrivial zero subset sum and gives a one-term
solution. Otherwise not all $w_{ij}$
are empty for
$i,j\leq \lfloor d/2 \rfloor$ and $\sum_{i,j=1}^{\lfloor d/2
\rfloor}iw_{ij}=0$.
\end{proof}

Iterated application of the method of Lemma~\ref{lem:halving}
gives the following result.

\begin{proposition}
\label{prop:signed}
Given $S_{\pm}(p,n)=p^{O(p\log p)}n^{O(p)}$ vectors
from $\Z_p^n$, 
a nontrivial signed subset sum representing the
zero vector can be found in deterministic time
$\poly(S_{\pm}(p,n))$.
\end{proposition}

\begin{proof}
Put $d_0=\frac{p-1}{2}$, $h_0(n)=n+1$
and define  
$d_i=\lfloor d_{i-1}/2\rfloor$ and 
$h_i(n)=h_{i-1}(n)h_{i-1}(\lceil d_{i-1}/2\rceil n)$
recursively for $i=1,\ldots,\floor{\log d_0}$.
As among any $h_0(n)=n+1$ vectors from $\Z_p^n$
a nontrivial linear relation can be found
in time $\poly(n\log p)$, recursive
applications of Lemma~\ref{lem:halving} 
gives that among $h_{\floor{\log d_0}}(n)$ vectors
a single nontrivial signed subset sum 
(that is, a 
linear relation
with nonzero coefficients $\pm 1$ 
only)
can be found in time $\poly(h_{\floor{\log d_0}}(n)\log p)$.
We show by induction that 
\begin{equation}
\label{ineq:complicated}
h_i(n)\leq \left(\prod_{j=0}^{i-1}\lceil d_j/2\rceil\right)
^{2^{i-1}}(n+1)^{2^i}.
\end{equation}
For $i=0$, both sides are equal to $n+1$.
Assume that the inequality holds for $0\leq i <\floor{\log d_0}$.
Then we also have
\begin{equation*}
h_i(\lceil d_i/2\rceil n)\leq \left(\prod_{j=0}^{i-1}\lceil d_j/2\rceil\right)
^{2^{i-1}}(\lceil d_i/2\rceil n+1)^{2^i}.
\end{equation*}
Using $\lceil d_i/2\rceil n+1\leq \lceil d_i/2\rceil(n+1)$,
we obtain
\begin{equation}
\label{ineq:morecomplicated}
h_i(\lceil d_i/2\rceil n)\leq \lceil d_i/2\rceil^{2^{i-1}}
\left(\prod_{j=0}^{i}\lceil d_j/2\rceil \right)
^{2^{i-1}}(n+1)^{2^i}.
\end{equation}
Multiplying inequalities (\ref{ineq:complicated})
and (\ref{ineq:morecomplicated}) 
and using $h_{i+1}(n)=h_i(n)h_i(\lceil d_i/2\rceil n)$,
we obtain
\begin{equation*}
h_{i+1}(n)
\leq \left(\prod_{j=0}^{i}\lceil d_j/2\rceil\right)
^{2^{i}}(n+1)^{2^{i+1}},
\end{equation*}
 which is
inequality (\ref{ineq:complicated}) for $i+1$
in place of $i$.
Using $d_j\leq d_0/2^j\leq d_0=\frac{p-1}{2}$,
inequalities (\ref{ineq:complicated}) for
$i=\floor{\log d_0}$ gives

$$h_{\floor{\log d_0}}\leq \left(\prod_{j=0}^{\left\lceil \log \frac{p-1}{2} \right\rceil-1}\left\lceil \frac{p-1}{4}\right\rceil\right)
^{2^{\left\lceil \log \frac{p-1}{2} \right\rceil-1}}(n+1)^{2^{\left\lceil \log \frac{p-1}{2} \right\rceil}}$$

$$=\left\lceil \frac{p-1}{4}\right\rceil^{\left\lceil\frac{p-1}{4}\right\rceil \left\lceil\log \frac{p-1}{2} \right\rceil} (n+1)^{\left\lceil\frac{p-1}{2}\right\rceil} .$$
Therefore, we have
$$h_{\floor{\log d_0}}(n)=p^{O(p\log p)}n^{O(p)}.$$
\end{proof}

We interpret a non-empty zero sum signed subset
as a non-trivial collision between two disjoint
subset sums. (Non-trivial means that at most one
of the subsets can be empty.) We use the short term 
{\em collision} for such a pair. We have the following.

\begin{proposition}
\label{prop:coll_doubling}
Suppose that there is an algorithm $\cal B$ that, given a set of
vectors from $\Z_p^n$ of size $S_\pm(p,n)$ finds 
a collision. Then there is a 
deterministic 
procedure that, given ${S_\pm(p,n)}^{\lceil \log p \rceil}$
vectors, finds a nontrivial zero sum subset using less than
$S_\pm^{\lceil \log p\rceil}$ 
applications of algorithm $\cal B$ and
$\poly((S_\pm(p,n))^{\lceil \log p \rceil})$ other operations.
\end{proposition}

\begin{proof}
Put $S=S_\pm(p,n)$ and $\ell=\lceil \log p \rceil$. 
We start with finding a collision $(H_1^+,H_1^-)$ among the first
$S$ vectors with common sum $w_1$ using algorithm $\cal B$. 
We continue with the next $S$ input vectors and find 
a
collision $(H_2^+,H_2^-)$ with sum $w_2$, and so on.
We then take the first $S$ subset sums $w_1,\ldots,w_S$ and
find a pair of disjoint subsets $(K^+,K^-)$ of $\{1,\ldots,S\}$,
not both empty, such that $\sum_{i\in K^+}w_i=\sum_{i\in K^-}w_i=w$.
The four subsets $L^{++}=\bigcup_{i\in K^+}H_i^+$,
$L^{+-}=\bigcup_{i\in K^+}H_i^-$, $L^{-+}=\bigcup_{i\in K^-}H_i^+$, and
$L^{--}=\bigcup_{i\in K^-}H_i^-$ of input vectors are pairwise disjoint,
not all empty and have common sum $w$. Iterating this we end up with 
at least $p$ pairwise disjoint subsets (not all empty) with equal
sum. If one of these sets is empty then the common sum is zero
and we can take any of the non-empty subsets. Otherwise the
union of the first $p$ of the subsets has zero sum.
The total number of applications of the collision finding algorithm $\cal B$
is $S^{\ell-1}+\ldots+S+1<S^\ell$.
\end{proof}

Propositions~\ref{prop:signed} and~\ref{prop:coll_doubling},
together with the remark on the case $p=2$
immediately give the following.

\begin{theorem}
\label{thm:zsum}
There is a deterministic algorithm that, given a sequence
of $S(p,n)=p^{O(p\log^2p)}n^{O(p\log p)}$ 
vectors 
from $\Z_p^n$, finds a 
non-trivial zero sum subsequence
in time 
$\poly(S(p,n))$. 
\end{theorem}

We remark that in \cite{IS17}, an algorithm for a  more general task
is given. This task is finding
a nontrivial representation of the zero vector as a linear
combination of the input vectors with $d$th power coefficients.
This includes our problem as the special case $d=p-1$. The
algorithm of \cite{IS17} for $d=p-1$ would give 
$S(p,n)=p^{O(p^2\log p)}n^{O(p\log p)}$, a parameter 
somewhat worse than that we have in Theorem~\ref{thm:zsum},
though would be still polynomial in $n$ for $p=O(1)$.
The method of \cite{IS17} for finding a collision 
is more complicated than the present one: it 
is based on collecting
relations organized in a $d$-dimensional hypercube
rather than a square. (The method of doubling collisions
is essentially identical with that described here in 
Proposition~\ref{prop:coll_doubling}.)

\section{Concluding remarks}
\label{sec:conclusion}

We have shown that the hidden subgroup problem in
a nilpotent group $G$ of class bounded by a constant
can be solved in polynomial time by an exact quantum
algorithm provided that there is a polynomial
time (that is, time $\poly(n\log p)$)
exact method that finds zero sum subsequences 
in sequences consisting of polynomially many elements of
$\Z_p^n$ for prime divisors $p$ of $\size{G}$. We have such
a method for $p=O(1)$. By Olson's theorem~\cite{Olson}, the 
shortest length for which a not necessarily polynomial time
 zero sum subsequence finding algorithm exists is around $np$. 

We propose the question of existence of a $\poly(np)$-time
algorithm for finding zero sum subsequences from sequences of
length $(np)^d$ for a sufficiently large constant $d$ as a problem
for further research.
A positive answer would imply existence of an exact polynomial
time quantum algorithm for the case when $\size{G}$ is smooth,
that is, the prime factors of $\size{G}$ are of size bounded
by a polynomial in $\log \size{G}$. Even a non-exact method
(e.g., a randomized algorithm) would be of great interest as,
by Proposition~\ref{prop:mod_comm_nonexact},
it would give a new result in the non-exact setting:
existence of an efficient "probabilistic" quantum hidden
subgroup algorithm for nilpotent groups of smooth order
having $O(1)$-bounded nilpotency class. Even somewhat worse
results would potentially lead to quantum hidden subgroup 
algorithms faster than the known ones.

For the purposes of "probabilistic" quantum hidden subgroup
algorithms even a method that finds a zero sum subsequence
"on average", that is for at least a $1/\poly(np)$ proportion of the possible
sequences would be sufficient. However, as the following simple
worst-case to average-case reduction shows, at least in the
randomized setting, the gain cannot
be better than polynomial. Assume that the classical randomized 
algorithm $\cal A$ finds in time
$T=T(p,n)$ with probability at least $\delta$ a subsequence
of a random sequence of length $S=S(p,n)$ of vectors from
$\Z_p^n$. Here, probability is taken for the uniform distribution
of the array of the vectors together with
the random bits of $\cal A$. 
Then we can do the following.
We start with an arbitrary sequence of $\frac{1}{\delta}\cdot S^2$ 
input vectors, we draw 
$\frac{1}{\delta}\cdot S^2$ uniformly random vectors, one for each input vector. 
Then we divide the input sequence into groups of length $S$
and to each input vector we add the corresponding random vector.
Within each group, we apply procedure $\cal A$. As the sums are random
vectors, in each group, procedure $\cal A$ succeeds with probability
at least $\delta$ and, with probability at least $\frac{1}{2}$,
$\cal A$ will succeed in at least $S$ groups. If this is the case then we choose
$S$ "lucky" groups, in each group take the sum of the random vectors corresponding to the
members of the zero sum subsequences. We apply algorithm $\cal A$ for
these $S$ sums. It finds a nontrivial zero sum subsequence with probability
at least $\delta$. Finally, we take the union of the corresponding 
subsequences. This way we obtain a procedure that finds
a nontrivial zero sum subsequence of {\em every} sequence
of length $\frac{1}{\delta}\cdot S^2$ in time  
$\poly(T+\frac{1}{\delta}\cdot ST)$ with probability at least
$\delta/2$.

\bibliographystyle{alpha}
\bibliography{myrefs}

\newcommand{\etalchar}[1]{$^{#1}$}
\begin{thebibliography}{WQT{\etalchar{+}}22}

\bibitem[AMR07]{AlaMooRus}
G.~Alagic, C.~Moore, and A.~Russell.
\newblock Quantum algorithms for simon's problem over general groups.
\newblock In {\em Proceedings of the eighteenth annual ACM-SIAM symposium on
  Discrete algorithms}, pages 1217--1224, 2007.

\bibitem[AR17]{AlaRus}
G.~Alagic and A.~Russell.
\newblock Quantum-secure symmetric-key cryptography based on hidden shifts.
\newblock In {\em Advances in Cryptology--EUROCRYPT 2017: 36th Annual
  International Conference on the Theory and Applications of Cryptographic
  Techniques, Paris, France, April 30--May 4, 2017, Proceedings, Part III},
  pages 65--93. Springer, 2017.

\bibitem[BH97]{BraHoy}
G.~Brassard and P.~H{\o}yer.
\newblock An exact quantum polynomial-time algorithm for {Simon}'s problem.
\newblock In {\em ISTCS 97}, pages 12--23, 1997.

\bibitem[BL21]{BaeLee}
E.~Bae and S.~Lee.
\newblock Quantum algorithm based on the $\varepsilon$-random linear
  disequations for the continuous hidden shift problem.
\newblock {\em Quantum Inf. Process.}, 20(10):347, 2021.

\bibitem[BNP18]{BonNaya}
X.~Bonnetain and M.~Naya-Plasencia.
\newblock Hidden shift quantum cryptanalysis and implications.
\newblock In {\em Advances in Cryptology--ASIACRYPT 2018: 24th International
  Conference on the Theory and Application of Cryptology and Information
  Security, Brisbane, QLD, Australia, December 2--6, 2018, Proceedings, Part I
  24}, pages 560--592. Springer, 2018.

\bibitem[Bon21]{BonnetainSimon}
X.~Bonnetain.
\newblock Tight bounds for {Simon}'s algorithm.
\newblock In {\em LATINCRYPT 2021}, pages 2--23, 2021.

\bibitem[BV97]{BerVaz}
E.~Bernstein and U.~Vazirani.
\newblock Quantum complexity theory.
\newblock {\em SIAM J. Comput.}, 26:1411--1473, 1997.

\bibitem[CM23]{CasVdm}
W.~Castryck and N.~Vander Meeren.
\newblock Two remarks on the vectorization problem.
\newblock In {\em Progress in Cryptology--INDOCRYPT 2022: 23rd International
  Conference on Cryptology in India, Kolkata, India, December 11--14, 2022,
  Proceedings}, pages 658--678. Springer, 2023.

\bibitem[CQ18]{CaiQiu}
G.~Cai and D.~Qiu.
\newblock Optimal separation in exact query complexities for {Simon}'s problem.
\newblock {\em J. Comput. Syst. Sci.}, 97:83--93, 2018.

\bibitem[DHIS14]{DHIS14}
T.~Decker, P.~H{\o}yer, G.~Ivanyos, and M.~Santha.
\newblock Polynomial time quantum algorithms for certain bivariate hidden
  polynomial problems.
\newblock {\em Quantum Inf. Comput.}, 14(9-10):790--806, 2014.

\bibitem[EHK04]{EHK04}
M.~Ettinger, P.~H{\o}yer, and E.~Knill.
\newblock The quantum query complexity of the hidden subgroup problem is
  polynomial.
\newblock {\em Inf. Process. Lett.}, 91:43--48, 2004.

\bibitem[FIM{\etalchar{+}}14]{FIMSS}
K.~Friedl, G.~Ivanyos, F.~Magniez, M.~Santha, and P.~Sen.
\newblock Hidden translation and translating coset in quantum computing.
\newblock {\em SIAM J. Comput.}, 43(1):1--24, 2014.
\newblock Preliminary version in STOC 2003.

\bibitem[HK18]{HorKah}
K.~Horan and D.~Kahrobaei.
\newblock The hidden subgroup problem and post-quantum group-based
  cryptography.
\newblock In {\em Mathematical Software--ICMS 2018: 6th International
  Conference, South Bend, IN, USA, July 24-27, 2018, Proceedings 6}, pages
  218--226. Springer, 2018.

\bibitem[II22]{II22}
M.~Imran and G.~Ivanyos.
\newblock An exact quantum hidden subgroup algorithm and applications to
  solvable groups.
\newblock {\em Quantum Inf. Comput.}, 22(9{\&}10):770--789, 2022.

\bibitem[IS17]{IS17}
G.~Ivanyos and M.~Santha.
\newblock Solving systems of diagonal polynomial equations over finite fields.
\newblock {\em Theor. Comput. Sci.}, 657:73--85, 2017.

\bibitem[ISS12]{ISS12}
G.~Ivanyos, L.~Sanselme, and M.~Santha.
\newblock An efficient quantum algorithm for the hidden subgroup problem in
  nil-2 groups.
\newblock {\em Algorithmica}, 62(1-2):480--498, 2012.

\bibitem[KL90]{KantorLux}
W.~M. Kantor and E.~M. Luks.
\newblock Computing in quotient groups.
\newblock In {\em STOC 1990}, pages 524--534, 1990.

\bibitem[Kup05]{Kuperberg}
G.~Kuperberg.
\newblock A subexponential-time quantum algorithm for the dihedral hidden
  subgroup problem.
\newblock {\em SIAM J. Comput.}, 35(1):170--188, 2005.

\bibitem[LK07a]{LomKaufGrover}
S.~Jr. Lomonaco and L.~H. Kauffman.
\newblock Is grover’s algorithm a quantum hidden subgroup algorithm?
\newblock {\em Quantum Inf. Process.}, 6:461--476, 2007.

\bibitem[LK07b]{LomKaufTool}
S.~Jr. Lomonaco and L.~H. Kauffman.
\newblock Quantum hidden subgroup algorithms: An algorithmic toolkit.
\newblock In {\em Mathematics of Quantum Computation and Quantum Technology},
  pages 21--64. Chapman and Hall/CRC, New York, 2007.

\bibitem[Lom04]{LomontSurvey}
C.~Lomont.
\newblock The hidden subgroup problem-review and open problems.
\newblock Technical Report quant-ph/0411037, arXiv, 2004.

\bibitem[MZ03]{MosZal}
M.~Mosca and C.~Zalka.
\newblock Exact quantum {Fourier} transforms and discrete logarithm algorithms.
\newblock {\em Int. J. Quantum Inf.}, 02:91--100, 2003.

\bibitem[NC10]{NiChu}
M.~A. Nielsen and I.~K. Chuang.
\newblock {\em Quantum Computation and Quantum Information}.
\newblock Cambridge University Press, Cambridge, 10th anniversary edition,
  2010.

\bibitem[NO02]{NishOzaQTM}
H.~Nishimura and M.~Ozawa.
\newblock Computational complexity of uniform quantum circuit families and
  quantum turing machines.
\newblock {\em Theor. Comput. Sci.}, 276:147--181, 2002.

\bibitem[NO05]{NishOzaUnif}
H.~Nishimura and M.~Ozawa.
\newblock Uniformity of quantum circuit families for error-free algorithms.
\newblock {\em Theor. Comput. Sci.}, 332:487--496, 2005.

\bibitem[NO09]{NishOzaPerfect}
H.~Nishimura and M.~Ozawa.
\newblock Perfect computational equivalence between quantum {Turing} machines
  and finitely generated uniform quantum circuit families.
\newblock {\em Quantum Inf. Process.}, 8:13--24, 2009.

\bibitem[Ols69]{Olson}
J.~E. Olson.
\newblock A combinatorial problem on finite abelian groups, i.
\newblock {\em J. Number Theory}, 1(1):8--10, 1969.

\bibitem[Rob95]{Robinson}
D.~Robinson.
\newblock {\em A Course in the Theory of Groups}.
\newblock Springer, Urbana, Illionis, 2nd edition, 1995.

\bibitem[Wan10]{WangSurvey}
F.~Wang.
\newblock The hidden subgroup problem.
\newblock Technical Report quant-ph/1008.0010, arXiv, 2010.

\bibitem[Wat01]{Watrous}
J.~Watrous.
\newblock Quantum algorithms for solvable groups.
\newblock In {\em STOC 2001}, pages 60--67, 2001.

\bibitem[WQT{\etalchar{+}}22]{WuQiuTanLiCai}
Z.~Wu, D.~Qiu, J.~Tan, H.~Li, and G.~Cai.
\newblock Quantum and classical query complexities for generalized {Simon}'s
  problem.
\newblock {\em Theor. Comput. Sci.}, 924:171--186, 2022.

\end{thebibliography}

\end{document}